\documentclass[runningheads]{llncs}

\usepackage{amsmath,amssymb,mathtools}
\usepackage{url}
\usepackage{xcolor}
\usepackage{tikz}
\usetikzlibrary{calc,positioning}
\usetikzlibrary{arrows.meta}
\usepackage[ruled,vlined,linesnumbered]{algorithm2e}

\SetKw{KwContinue}{continue}
\SetKw{KwRet}{return}
\usepackage{microtype}

\usepackage{orcidlink}
\hypersetup{hidelinks}
\providecommand{\orcidID}[1]{}
\renewcommand{\orcidID}[1]{\orcidlink{#1}}

\newcommand{\eps}{\varepsilon}
\newcommand{\RLE}{\mathsf{RLE}}
\newcommand{\RLEST}{\mathsf{RLEST}}
\newcommand{\RLESA}{\mathsf{RLESA}}
\newcommand{\RLELCP}{\mathsf{RLELCP}}
\newcommand{\lcp}{\mathsf{lcp}}

\newcommand{\lab}{\mathsf{lab}}
\newcommand{\pos}{\mathsf{pos}}
\newcommand{\pchar}{\mathsf{pchar}}
\newcommand{\plen}{\mathsf{plen}}
\newcommand{\per}{\mathsf{per}}
\newcommand{\expn}{\mathsf{exp}}
\newcommand{\spanstr}{\mathsf{span}}
\newcommand{\Lout}{\mathsf{Lout}}
\newcommand{\Rout}{\mathsf{Rout}}
\newcommand{\rev}[1]{#1}

\title{Compact Enumeration of Maximal Closed Substrings in Run-Length Encoded Strings}
\titlerunning{Compact Enumeration of MCSs in RLE Strings}

\author{
  Haruki~Umezaki\inst{1}
  \and
  Hiroki~Shibata\inst{2}\orcidID{0009-0006-6502-7476}
  \and
  Yuto~Nakashima\inst{3}\orcidID{0000-0001-6269-9353}
  \and 
  Shunsuke~Inenaga\inst{3}\orcidID{0000-0002-1833-010X}
}

\authorrunning{Umezaki, Shibata, Nakashima, Inenaga}

\institute{
  {Department of Information Science and Technology, Kyushu University} \\
  \email{umezaki.haruki.314@s.kyushu-u.ac.jp}
  \and
  {Joint Graduate School of Mathematics for Innovation, Kyushu University} \\
  \email{shibata.hiroki.753@s.kyushu-u.ac.jp}
  \and
  {Department of Informatics, Kyushu University} \\
  \email{\{nakashima.yuto.003, inenaga.shunsuke.380\}@m.kyushu-u.ac.jp} 
}

\begin{document}
\maketitle

\begin{abstract}
A string $w$ is \emph{closed} if $|w|=1$, or if $w$ has a non-empty proper border occurring only as its prefix and suffix.
A \emph{maximal closed substring} (\emph{MCS}) is a maximal occurrence of a closed string; equivalently, it is a \emph{maximal closed repeat} (\emph{MCR}).
We study the problem of enumerating all MCS occurrences directly from the \emph{run-length encoding} (\emph{RLE}) of a string.
For a string $T$ of length $n$ with RLE size $m$, we give a compact representation of all MCS occurrences whose worst-case size is $O(m^2)$, and show that this bound is tight for this representation.
Our approach is based on a characterization of MCS occurrences in terms of consecutive occurrences of their longest borders, together with data structures built on the RLE of $T$.
Denoting the resulting representation by $\mathcal{F}$, we compute it in
$O(m\log^2 m + |\mathcal{F}|\log m)$ time using $O(m)$ working space.
\end{abstract}

\section{Introduction}

Repeated substrings are among the most fundamental objects in string processing. Classical examples include squares, runs, maximal repeats, gapped repeats, and their algorithmic variants. A string is called \emph{closed} if it is of length one, or if it has a non-empty proper border that occurs in the string only as the prefix and the suffix. Badkobeh et al.~\cite{badkobeh2022mcs} introduced \emph{maximal closed substrings} (\emph{MCSs}), namely occurrences of closed substrings that cannot be extended to the left nor to the right into a longer closed substring. 
They showed that MCSs of length at least twice their smallest period are
exactly maximal repetitions (also called runs), while the other MCSs form
a subclass of maximal gapped repeats.
They also showed that all MCSs in a string of length $n$ can be computed in $O(n\log n)$ time with $O(n)$ space.
Their algorithm also implies an $O(n \log n)$ upper bound for the number of MCSs occurring in any string of length $n$~\cite{badkobeh2022mcs}.
Kosolobov~\cite{kosolobov2024closedrepeats} later introduced \emph{maximal closed repeats} (\emph{MCRs}) and showed a matching $\Omega(n \log n)$ lower bound on their number in a string of length $n$.

Closed strings have also been studied from the broader combinatorial viewpoint.
Badkobeh et al.~\cite{badkobeh2015closedfactors} proved, among other
results, that a length-$n$ string may contain $\Theta(n^2)$ distinct \rev{closed
substrings}.  More recently, Jain and Mhaskar~\cite{jain2025closedsubstrings} gave
a compact representation of all closed substrings and an alternative algorithm
for computing MCSs.  Their compact representation is for the uncompressed
setting; in contrast, the RLE-based compact representation studied here is
measured in terms of the RLE size $m$.  On the repeat side, non-periodic MCSs are closely related
to maximal gapped repeats and maximal pairs, whose combinatorial bounds and
enumeration algorithms have been studied extensively; see, e.g.,
Brodal et al.~\cite{brodal1999maximalpairs} and Gawrychowski et
al.~\cite{gawrychowski2018gapped}.  Our compressed-input viewpoint belongs to
the complementary line of direct processing of RLE strings, which includes
maximal repetitions~\cite{fujishige2017maximal_rle}, minimal absent
words~\cite{akagi2022maw_rle}, and longest unbordered
factors~\cite{sekizaki2025unbordered_rle}.

Motivated by this line of direct RLE processing, we ask whether MCS occurrences
can likewise be represented and enumerated without decompressing the input.
Given a string $T$ whose RLE $c_1^{e_1}\cdots c_m^{e_m}$ has size $m$,
with $c_i \ne c_{i+1}$ and $e_i \ge 1$, our goal is to describe and enumerate
the MCS occurrences of $T$ directly from its RLE.
The number of MCS occurrences, however, is not bounded by any function of $m$
alone.
For example, for $T=a^kba^k$, the RLE size is $m=3$, while
$a^i b a^i$ is an MCS occurrence for every $1\le i\le k$.
Thus, even for constant RLE size, the number of MCS occurrences can be
arbitrarily large.
We therefore seek a compact representation whose size is bounded in terms of
$m$.

Our starting point is a simple characterization of an MCS occurrence.
Consider an MCS occurrence $W$ of length greater than one, and let $X$ be its
longest border.
Since $X$ occurs in $W$ only as its prefix and suffix, these two occurrences
of $X$ are consecutive occurrences in $T$: there is no occurrence of $X$
starting strictly between them.
Conversely, two consecutive occurrences of a string $X$ determine a closed
substring spanning them, provided that the two occurrences satisfy the appropriate
left- and right-maximality conditions.
We call such a pair a \emph{valid pair}.
Thus, enumerating MCS occurrences can be reduced to enumerating valid pairs of
occurrences of their longest borders.
Our algorithms exploit the RLE structure by classifying valid pairs according
to the RLE length of the border.
The rest of this paper establishes the following result.

\begin{theorem}\label{thm:main}
  Given an RLE of size $m$ representing string $T$,
  a compact representation $\mathcal{F}$ of size $O(m^2)$
  that represents all MCS occurrences in $T$
  can be computed in $O(m\log^2 m+|\mathcal F|\log m)$ time with
  $O(m)$ working space.
\end{theorem}

\section{Preliminaries}

Let $T$ be a string of length $n=|T|$ over an ordered alphabet $\Sigma$, and let $\sigma$ be the number of distinct characters occurring in $T$. We use the standard word-RAM model with word size $w=\Omega(\log n)$ and assume that each alphabet symbol fits in one machine word; arithmetic and comparisons on positions, run lengths, and alphabet symbols take $O(1)$ time. The notation $T[i]$ denotes the $i$th character of $T$ for $1 \le i \le n$. For $1 \le i \le j \le |T|$, let $T[i..j]$ denote the substring beginning at position $i$ and ending at position $j$.
We use the convention that $T[i..j]=\eps$ if $i>j$.
The substrings $T[1..j]$ and $T[i..|T|]$ are called prefixes and suffixes of $T$, respectively.
For two strings $X$ and $Y$,
let $\lcp(X, Y)$ denote the length of the longest common prefix of $X$ and $Y$.

A string $X$ is a \emph{border} of a string $W$ if $X$ is both a prefix and a suffix of $W$. A border is \emph{proper} if $X\ne\eps$ and $X\ne W$. A positive integer $p\le |W|$ is a \emph{period} of $W$ if $W[t]=W[t+p]$ for every $1\le t\le |W|-p$. For a non-empty string $W$, $\per(W)$ denotes its smallest period and $\expn(W)=|W|/\per(W)$ its exponent.

A string $W$ is \emph{closed} if $|W|=1$, or if there is a non-empty proper border $X$ of $W$ such that $X$ occurs in $W$ only as the prefix and the suffix. An occurrence $T[p..q]$ of a closed string is a \emph{maximal closed substring occurrence}, or \emph{MCS occurrence}, if $p=1$ or $T[p-1..q]$ is not closed, and $q=|T|$ or $T[p..q+1]$ is not closed~\cite{badkobeh2022mcs}.
A length-one occurrence $T[p..p]$ is an MCS if and only if each existing neighboring character differs from $T[p]$; such occurrences can be reported trivially.
Hence, in what follows, we consider only MCS occurrences $W$ with $|W|>1$.

Let $W$ be a closed string of length at least two, and let $X$ be the longest border of $W$.
Then, $X$ occurs in $W$ only as the prefix and the suffix.

For a non-empty string $X$, let
$p_1, \ldots, p_s$
be the list of starting positions of occurrences of $X$ in $T$
arranged in increasing order.
A pair $(p_t,p_{t+1})$ is called a \emph{consecutive occurrence pair} of $X$
for $1 \leq t < s$.
\begin{definition}[valid pair]\label{def:valid-pair}
Let $X$ be a non-empty string and let $i<j$ be two occurrences of $X$ in $T$. The triple $(X,i,j)$ is a \emph{valid pair} if
\begin{enumerate}
\item $i$ and $j$ form a consecutive occurrence pair of $X$;
\item $i=1$ or $T[i-1]\ne T[j-1]$;
\item $j+|X|-1=|T|$ or $T[i+|X|]\ne T[j+|X|]$.
\end{enumerate}
The substring spanned by $(X,i,j)$ is $\spanstr(X,i,j)=T[i..j+|X|-1]$.
We refer to Conditions~2 and~3 as the \emph{left-maximality} and \emph{right-maximality} conditions, respectively.
\end{definition}

The spanned substring forms an MCS.
Moreover, we have:
\begin{lemma}[Adapted from~\cite{badkobeh2022mcs,kosolobov2024closedrepeats}]\label{lem:valid-pair-mcs}
  There is a one-to-one correspondence between MCS occurrences of length greater than one and valid pairs, where an MCS occurrence is represented by its longest border and the two occurrences of this border as the prefix and the suffix.
\end{lemma}

\begin{definition}[horizontal visibility]\label{def:horizontal-visibility}
Let $P$ be a finite set of points in the plane with distinct $x$-coordinates. For each $p\in P$, let $x(p)$ denote its $x$-coordinate and let $h(p)$ denote its height (its $y$-coordinate). Two points $q,r\in P$ with $x(q)<x(r)$ are
\emph{visible} if
$\max\{h(p): x(q)<x(p)<x(r)\} < \min\{h(q),h(r)\}$.
The maximum over the empty set is defined as $-\infty$.  For a sequence
$h_1,\ldots,h_k$, its horizontal visibility graph is obtained by applying this
definition to the points $(1,h_1),\ldots,(k,h_k)$.
\end{definition}

This is the standard horizontal visibility graph introduced by
Luque et al.~\cite{luque2009hvg}.
See Fig.~\ref{fig:horizontal-visibility-graph} for a concrete example
of horizontal visibility graphs.
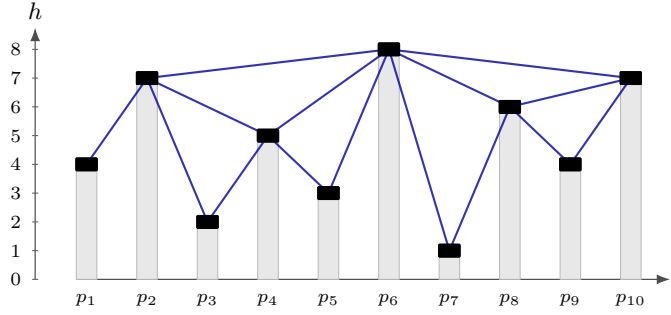
\begin{figure}[btph]
\centering
\begin{tikzpicture}[
  x=0.80cm,
  y=0.38cm,
  >=Latex,
  bar/.style={fill=gray!18, draw=gray!55, line width=0.35pt},
  point/.style={fill=black, draw=black, rounded corners=0.4pt},
  visedge/.style={draw=blue!65!black, line width=0.80pt, opacity=0.80},
  scaleaxis/.style={draw=black!70, line width=0.45pt},
  tick/.style={draw=black!65, line width=0.35pt}
]

\def\barw{0.34}

\draw[scaleaxis,->] (-0.85,0) -- (-0.85,8.75) node[above] {$h$};
\foreach \t in {0,1,...,8}{
  \draw[tick] (-0.90,\t) -- (-0.80,\t);
  \node[anchor=east, font=\scriptsize] at (-0.94,\t) {$\t$};
}

\draw[scaleaxis,->] (-0.15,0) -- (9.65,0);

\foreach \i/\h in {1/4,2/7,3/2,4/5,5/3,6/8,7/1,8/6,9/4,10/7}{
  \pgfmathsetmacro{\x}{\i-1}
  \filldraw[bar] ({\x-\barw/2},0) rectangle ({\x+\barw/2},\h);
  \node[point, minimum width=8pt, minimum height=4.8pt, inner sep=0pt]
    (p\i) at (\x,\h) {};
  \node[below, font=\scriptsize] at (\x,-0.22) {$p_{\i}$};
}

\foreach \i/\j in {1/2,2/3,2/4,2/6,3/4,4/5,4/6,5/6,6/7,6/8,6/10,7/8,8/9,8/10,9/10}{
  \draw[visedge] (p\i.center) -- (p\j.center);
}

\foreach \i/\h in {1/4,2/7,3/2,4/5,5/3,6/8,7/1,8/6,9/4,10/7}{
  \pgfmathsetmacro{\x}{\i-1}
  \node[point, minimum width=8pt, minimum height=4.8pt, inner sep=0pt]
    at (\x,\h) {};
}

\end{tikzpicture}
\caption{An example of a horizontal visibility graph. The histogram bar at each position $p_i$ has height $h_i$; two points are connected when no intermediate bar reaches height $\min\{h_i,h_j\}$. Blue edges are drawn as straight line segments connecting the black endpoint vertices.}
\label{fig:horizontal-visibility-graph}
\end{figure}

Horizontal visibility graphs are known to be outerplanar~\cite{gutin2011horizontal}.
Recall that a graph is outerplanar if it admits a plane embedding in which
all vertices lie on the boundary of the outer face.
We use the following standard property, for which we give a proof for completeness.

\begin{lemma}\label{lem:hvg-linear}
The horizontal visibility graph with $k \geq 2$ \rev{vertices} has at most $2k-3$ edges.
\end{lemma}

\begin{proof}
For $k=2$, the claim is immediate.  For $k\ge3$,
let the points be $p_1,\ldots,p_k$ in increasing order of their
$x$-coordinates, and write $h_i$ for the height of $p_i$.
We first observe that no two edges joining horizontally visible pairs cross.
Suppose, to the contrary, that both $p_i,p_j$ and $p_{i'},p_{j'}$
are horizontally visible with $i<i'<j<j'$.
Since $p_i$ and $p_j$ are horizontally visible, we have
$h_{i'} < \min\{h_i,h_j\} \le h_j$.
On the other hand, since $p_{i'}$ and $p_{j'}$ are horizontally visible,
$h_j < \min\{h_{i'},h_{j'}\} \le h_{i'}$, a contradiction.
Thus, placing the vertices in the order
$p_1,\ldots,p_k$ on a line and drawing an edge between every horizontally visible pair as an arc above the line gives an outerplanar embedding.
Any outerplanar graph on $k \ge 2$ vertices can be augmented to a
maximal outerplanar graph, which has exactly $2k-3$ edges.
Hence the horizontal visibility graph has at most $2k-3$ edges.
\qed
\end{proof}

\section{Sparse Suffix Tree on RLE}\label{sec:sparse-st}

The \emph{run-length encoding} (\emph{RLE}) of a string $T$ is its unique
factorization $\RLE(T)=c_1^{e_1}\cdots c_m^{e_m}$, where $e_i\ge 1$ and
$c_i\ne c_{i+1}$ for $1\le i<m$.  Each factor $c_i^{e_i}$ is a \emph{run}, and
we write $\rho(T)=m$.  Let
$S_i=1+\sum_{k<i}e_k$ and $E_i=S_i+e_i-1$ be the starting and ending positions
of the $i$th run, and let $\mathcal S_T=\{S_i:1\le i\le m\}$.
For $2\le i\le m$, the cut between positions $E_{i-1}$ and $S_i$ is called an
\emph{RLE boundary} (or \emph{run boundary}); we identify this boundary with
its right-hand position $S_i$.  Thus the set of RLE boundaries is
$\mathcal B_T=\{S_i:2\le i\le m\}$.

Let $\$\notin\Sigma$ be a fresh terminal symbol.  The compact trie of the
terminal-marked run-start suffixes $\{T[S_i..n]\$:S_i\in\mathcal S_T\}$ is
called the \emph{RLE suffix tree} $\RLEST(T)$; edge lengths are measured in the
expanded string.  The \emph{RLE suffix array} $\RLESA(T)$ lists the positions
$S_i$ in lexicographic suffix order, and $\RLELCP(T)$ stores the LCP lengths of
adjacent suffixes in this order.  Tamakoshi et al.~\cite{tamakoshi2015rle_index}
showed that these arrays can be constructed in $O(m\log m)$ time and $O(m)$
space; $\RLEST(T)$ is then obtained in $O(m)$ time by the standard suffix-tree
construction from a suffix array and an LCP array.

For a node $v$ of $\RLEST(T)$, let
$\lab(v)$ denote its root-to-$v$ path-label.  The \emph{locus} of a non-empty
prefix $Y$ of a run-start suffix is the unique point at string depth $|Y|$; it is
\emph{explicit} if this locus is a node.

For each leaf $\ell_i$ corresponding to $T[S_i..|T|]\$$, we store
$\pos(\ell_i)=S_i$.  If $i>1$, we also store the character and length of the
preceding run, namely $\pchar(\ell_i)=c_{i-1}$ and $\plen(\ell_i)=e_{i-1}$; for
$i=1$ we set $\pchar(\ell_i)=\bot$ and $\plen(\ell_i)=0$, where
$\bot\notin\Sigma$ is a special symbol.  These values are obtained directly from the RLE representation.

\section{Classification and Family-Size Bounds}\label{sec:classification}

\begin{definition}[compact family representation]\label{def:family-representation}
We identify a substring occurrence $T[\ell..r]$ with its interval
$(\ell,r)$.  A \emph{family} is one of the following two types:
\begin{enumerate}
\item A \emph{singleton family} $\mathsf{S}(\ell,r)$ represents the
      single occurrence $(\ell,r)$.
\item Let $r<s$ be run indices with $c_r=c_s=a$, and let
      $J=[L,R]\subseteq[1..\min(e_r,e_s)]$ be a non-empty integer interval.
      A \emph{unary-border family} $\mathcal U_a(r,s,J)$ represents
      \[
        \{(E_r-t+1,\,S_s+t-1):t\in J\}.
      \]
\end{enumerate}
A collection $\mathcal F$ of such families \emph{represents all MCS
occurrences} if the union of the represented occurrence intervals is exactly
the set of all MCS occurrences of $T$.  We write $|\mathcal F|$ for the number
of families, not for the number of represented occurrences.  Each family can
be represented in $O(1)$ machine words.
\end{definition}

Recall that we consider only MCSs $W$ with $|W| > 1$.
Let $(X,i,j)$ be a valid pair, and let $W=\spanstr(X,i,j)$.
When $\expn(W)<2$, the two border occurrences do not overlap; hence
we can write $W=XUX$, where $U$ is the non-empty substring between them.
We classify MCS occurrences as follows.
\begin{definition}[types]\label{def:types}
  An MCS occurrence $W$ with longest border $X$ is assigned to one of the following five types. Whenever $\expn(W)<2$, we write $W=XUX$, where $U\ne\eps$ is the substring between the two non-overlapping occurrences of $X$.
\begin{description}
\item[Type 1]: $\rho(W)=1$.
\item[Type 2]: $\expn(W)<2$, $\rho(X)=1$, and $\rho(U)=1$.
\item[Type 3]: $\expn(W)<2$, $\rho(X)=1$, and $\rho(U)\ge 2$.
\item[Type 4]: $\expn(W)<2$ and $\rho(X)\ge 2$.
\item[Type 5]: $\rho(W)\ge 2$ and $\expn(W)\ge 2$.
\end{description}
\end{definition}

We call Types~2 and~3 the \emph{unary-border} cases, Type~4 the \emph{non-unary non-periodic} case, and Types~1 and~5 the periodic cases. This section analyzes the size of the compact representation for each type and also records the known RLE enumeration result for Types~1 and~5. Sections~\ref{sec:unary} and~\ref{sec:multi-run} enumerate Types~2--4.

\paragraph{Types 1 and 5.}
The MCS occurrences of Types 1 and 5 are exactly maximal repetitions~\cite{KolpakovK99}:
A substring $R=T[\ell..r]$ with smallest period $p$ is a maximal repetition if
$|R|\ge 2p$, and $\ell=1$ or $T[\ell-1]\ne T[\ell+p-1]$, and
$r=|T|$ or $T[r+1]\ne T[r-p+1]$.
Badkobeh et al.~\cite{badkobeh2022mcs} observed that MCSs with exponent at
least two are maximal repetitions.  Thus Types 1 and 5 can be computed by
algorithms for computing maximal repetitions on RLE strings.  Fujishige et
al.~\cite{fujishige2017maximal_rle} showed that all maximal repetitions in an
RLE string of size $m$ can be computed in $O(m\alpha(m))$ time and $O(m)$
space, where $\alpha$ denotes the inverse Ackermann function, and that their
number is $O(m)$.
Thus we have the following:
\begin{lemma}\label{lem:type5}
  The MCS occurrences of Type 1 and Type 5 contribute $O(m)$ singleton
  families, and can be computed in $O(m\alpha(m))$ time and $O(m)$ space.
\end{lemma}

\paragraph{Types 2 and 3.}
For unary-border MCSs, fix a character $a$.  Let $r_1<\cdots<r_k$ be the run indices in $\RLE(T)$ satisfying $c_{r_p}=a$, and put $h_p=e_{r_p}$. For such a run index $r$, let $\Lout(r)$ and $\Rout(r)$ be the characters immediately outside that run, using two distinct sentinels $\#_L,\#_R\notin\Sigma$ at the left and right ends of $T$, respectively.  For $p<q$, let
$M_{p,q}=\max_{p<s<q} \{h_s\}$, with $M_{p,q}=0$ if $q=p+1$.  The two boundary
occurrences of $a^t$ in the \rev{$p$-th and $q$-th $a$-runs}
are consecutive exactly for
$M_{p,q}<t\le \min\{h_p,h_q\}$.  Hence the admissible values of $t$ first form the
interval
$I_{p,q}=[M_{p,q}+1,\min\{h_p,h_q\}]$.  In the sense of
Definition~\ref{def:horizontal-visibility} applied to the height sequence
$h_1,\ldots,h_k$, the pair $(p,q)$ is visible exactly when this interval is
non-empty.

\begin{lemma}\label{lem:unary-family}
Let $I_{p,q}=[L,R]$ be non-empty. The values $t$ for which the two occurrences of $a^t$ in the \rev{$p$-th and $q$-th $a$-runs} form a valid pair form either $I_{p,q}$ or $I_{p,q}\setminus\{R\}$. If this set is non-empty, it is an interval $J$ represented by the unary-border family $\mathcal U_a(r_p,r_q,J)$ of Definition~\ref{def:family-representation}.
\end{lemma}

\begin{proof}
The interval condition is precisely the absence of an intermediate $a$-run of
length at least $t$.  If $t<R$, then $t<h_p,h_q$.  Thus the first occurrence
is preceded by $a$ and followed by a non-$a$ character, while the second is
preceded by a non-$a$ character and followed by $a$; hence both maximality
conditions hold.
Only $t=R$ can fail maximality: left maximality fails exactly when
$h_p=R$ and $\Lout(r_p)=\Lout(r_q)$, and right maximality fails exactly
when $h_q=R$ and $\Rout(r_p)=\Rout(r_q)$.
Hence, after possibly removing $R$, the \rev{remaining values of $t$} form an interval.
If this interval is non-empty, it is represented by the unary-border family
$\mathcal U_a(r_p,r_q,J)$.
\qed
\end{proof}

\begin{lemma}\label{lem:unary-number}
  All Type 2 and Type 3 MCS occurrences can be represented by $O(m)$
  unary-border families.
\end{lemma}

\begin{proof}
For each $a \in \Sigma$, non-empty intervals are exactly the edges of the
horizontal visibility graph of $h_1,\ldots,h_k$ from
Definition~\ref{def:horizontal-visibility}, which has $O(k)$ edges by
Lemma~\ref{lem:hvg-linear}.  By Lemma~\ref{lem:unary-family}, each such edge gives at most
one unary-border family.  Summing over all characters gives $O(m)$ families,
since the character classes partition the runs.
\qed
\end{proof}

\paragraph{Type 4.}
For the non-periodic cases with $\rho(X)\ge2$, singleton families are enough.

\begin{lemma}\label{lem:boundary-charge}
The number of valid pairs $(X,i,j)$ with $\rho(X)\ge 2$ is $O(m^2)$.
\end{lemma}

\begin{proof}
Let $(X,i,j)$ be such a valid pair and put $d=j-i$. Since $\rho(X)\ge 2$, the first occurrence of $X$ contains an RLE boundary corresponding to a character change within $X$. Let $\beta$ be the leftmost such boundary position, so that $i<\beta\le i+|X|-1$. The corresponding position $\beta+d$ is then an RLE boundary in the second occurrence. Charge the pair to the ordered boundary pair
$(\beta,\beta+d)$.

Fix an ordered boundary pair $(\beta,\beta')$, and let $d=\beta'-\beta$.  Let $L$ be the
maximum number of characters that match immediately to the left of $\beta$ and
$\beta'$, and let $R$ be the maximum number of characters that match starting at
$\beta$ and $\beta'$, respectively.  Thus the two aligned substrings agree exactly on
$T[\beta-L..\beta+R-1]=T[\beta'-L..\beta'+R-1]$.  If a valid pair is charged to $(\beta,\beta')$,
then Conditions~2 and 3 of Definition~\ref{def:valid-pair} force
$i=\beta-L$, $|X|=L+R$, and $j=i+d$.  Hence $(\beta,\beta')$ uniquely determines
$(X,i,j)$.  There are $O(m^2)$ ordered pairs of RLE boundaries.
\qed
\end{proof}

\begin{theorem}\label{thm:lower-bound}
There are RLE strings of size $m$ that contain $\Omega(m^2)$ Type 4 MCS occurrences.
\end{theorem}

\begin{proof}
Let $h\ge 1$. Consider a string consisting of $2h$ blocks
\[
T_h=(a^{L_1}b^{R_1}\delta_1)\cdots(a^{L_{2h}}b^{R_{2h}}\delta_{2h}),
\]
where all $\delta_p$ are distinct and different from $a,b$.
For $1\le i\le h$, let
$L_i=h-i+1$ and $R_i=h+1$.
For $1\le q\le h$, let
$L_{h+q}=h+1$ and $R_{h+q}=q$.
For each pair $(i,q)\in[1..h]^2$, let $X_{i,q}=a^{h-i+1}b^q$.

The string $X_{i,q}$ occurs at the $a|b$ boundary of block $i$ and at the
$a|b$ boundary of block $h+q$.  No intermediate block contains $X_{i,q}$:
among the later blocks of the first half, the $a$-run is too short, whereas
among the earlier blocks of the second half, the $b$-run is too short.
Thus these two occurrences are consecutive.  Conditions~2 and~3 of
Definition~\ref{def:valid-pair} also hold.  Indeed, at block $i$ the occurrence
starts at the left boundary of its $a$-run, whereas at block $h+q$ it starts
strictly inside the run, so the preceding characters differ.  Symmetrically,
at block $h+q$ it ends at the right boundary of its $b$-run, whereas at block
$i$ it ends strictly inside the run, so the following characters differ.

Let $W_{i,q}$ be the substring spanned by these two occurrences.  The separator
$\delta_i$ occurs in $W_{i,q}$ exactly once.  If $W_{i,q}$ had a period
$p\le |W_{i,q}|/2$, then at least one of the positions at distance $p$ from
this occurrence of $\delta_i$ would still lie inside $W_{i,q}$.  Periodicity
would force that position to contain $\delta_i$ as well, contradicting
uniqueness.  Therefore $\expn(W_{i,q})<2$.
Hence the $h^2$ pairs yield distinct Type 4 MCS occurrences, and since the RLE
size is $\Theta(h)$, the lower bound is $\Omega(m^2)$.  Moreover, every
$W_{i,q}$ begins with $a$ and ends with $b$, where $a\ne b$, whereas every
occurrence represented by a unary-border family begins and ends with the
same character.  Thus none of these $h^2$ occurrences can be represented by a
unary-border family, and each requires its own singleton family.
\qed
\end{proof}

For a concrete example of Theorem~\ref{thm:lower-bound},
consider the case where $h=3$.
The six blocks have run-length pairs
\[
\underbrace{(3,4),(2,4),(1,4)}_{\text{first half}}
\qquad
\underbrace{(4,1),(4,2),(4,3)}_{\text{second half}},
\]
that is,
\[
T_3=(a^3b^4\delta_1)(a^2b^4\delta_2)(ab^4\delta_3)
    (a^4b\delta_4)(a^4b^2\delta_5)(a^4b^3\delta_6).
\]
For $(i,q)=(2,3)$, the border $X_{2,3}=a^2b^3$ occurs at the
$a|b$ boundaries of blocks $2$ and $6$.  No intermediate block has both
enough $a$'s and enough $b$'s, while the first occurrence starts at the left boundary of its $a$-run but the second does not, and the second ends at the right boundary of its $b$-run but the first does not; this illustrates consecutiveness and maximality.

Combining Lemmas~\ref{lem:type5}, \ref{lem:unary-number}, and
\ref{lem:boundary-charge} with the $O(m)$ length-one MCS occurrences and
Theorem~\ref{thm:lower-bound}, we obtain the following.

\begin{theorem}\label{thm:family-size}
For any string $T$ of RLE size $m$, all MCS occurrences of $T$ can be
represented by $O(m^2)$ compact families. This bound is worst-case optimal.
\end{theorem}

Non-periodic MCSs of Types 2--4 have the form $W=XUX$, where $X$ is the
longest border and $U\ne\varepsilon$.  Sections~\ref{sec:unary}
and~\ref{sec:multi-run} enumerate the unary-border families and the remaining
non-unary Type~4 families, respectively.  The latter case uses horizontal
visibility (Definition~\ref{def:horizontal-visibility}) on point sets \rev{associated
with loci in the sparse suffix tree}.

\section{Enumeration of Non-Periodic MCSs with Unary Borders}\label{sec:unary}

We now enumerate the Type 2 and Type 3 families described in
Lemma~\ref{lem:unary-family}.  The proof of Lemma~\ref{lem:unary-number} is
constructive: for each character $a$, it suffices to list the visible pairs
(Definition~\ref{def:horizontal-visibility}) in the height sequence of the
$a$-runs and output the corresponding \rev{unary-border family after the
maximality test of Lemma~\ref{lem:unary-family}}.

\begin{lemma}\label{lem:unary-enum}
All unary-border families can be enumerated in $O(m\log\sigma)$ time with
$O(m)$ working space.  Since each distinct character occurring in $T$ labels at least one run, $\sigma\le m$.
\end{lemma}

\begin{proof}
Group the runs by their character in $O(m\log\sigma)$ time using a balanced
search tree.  For each character $a$, enumerate the horizontal visibility graph
of its run-length sequence (Definition~\ref{def:horizontal-visibility}) by a left-to-right monotone-stack scan.  Each run is pushed once and popped at most once, so the scan and the reported visibility edges take linear time in the number of $a$-runs.  For each visible pair, compute the interval
$J$ of Lemma~\ref{lem:unary-family} in constant time from the two endpoint run
lengths and outside characters, and output the family if $J$ is non-empty.
Since the character classes partition the runs, the total work after grouping
is $O(m)$ and the working space is $O(m)$.
\qed
\end{proof}

\section{Enumeration of Non-Periodic MCSs with Non-Unary Borders}\label{sec:multi-run}

Here, we enumerate Type~4 MCS occurrences $XUX$ whose longest border $X$
has $\rho(X)\ge 2$.  Write the border as $X=a^KY$, where $a^K$ is the \rev{first
run of $X$}.  Then $K\ge 1$, and the remaining suffix $Y$ is non-empty and
starts at a run boundary.

\paragraph{Overview.}
The enumeration has four steps.  First, fixing $Y$ places its occurrences at
one explicit locus $v$ of the sparse suffix tree, and the two endpoint
leaves of a \rev{valid pair} lie in distinct children of $v$.  Second, for
each preceding character $a$, all leaves below $v$ are mapped to points whose
horizontal coordinate is the occurrence position and whose height is the length
of the preceding $a$-run.  Horizontal visibility in this full point set is
exactly the absence of an intermediate occurrence of the candidate border
$a^KY$.  Third, a constant-time test checks left maximality and a length test
discards the periodic cases.  Finally, at every node $v$,
we keep a largest child and query all leaves in the other children against the
full point set of $v$.  Every queried leaf then belongs to a child containing at
most half of the leaves of $v$, which gives the same doubling argument as in
standard smaller-to-larger merging.  The remainder of this section formalizes
these steps and proves correctness and the output-sensitive bound.

\subsection{Representation at a Suffix-Tree Locus}

For any suffix-tree node $x$, let $L(x)$ denote the set of leaves in the subtree rooted at $x$.
Let $v$ be the locus of $Y$ in $\RLEST(T)$.  A leaf
$\ell\in L(v)$ represents an occurrence of $Y$ starting at a run boundary.
Besides $\pchar(\ell)$ and $\plen(\ell)$, we use $\Lout(\ell)$ for the character
immediately before the preceding run, or the left sentinel $\#_L$ at the left end of $T$.
This value is obtained in constant time from the run represented by $\ell$.

When $\rho(X)=2$, we have $\rho(Y)=1$.  This case is handled without any
special treatment: the leaf set of the locus of $Y$ consists exactly of the
run-start suffixes having $Y$ as a prefix, and hence automatically excludes
runs that are too short to contain $Y$.

Consider two leaves
$\ell,z\in L(v)$ lying below distinct children of $v$.  If
$\pchar(\ell)=\pchar(z)=a\ne\bot$, then they define the candidate border
$a^KY$ with $K=\min(\plen(\ell),\plen(z))$.  The candidate is left-maximal
exactly when $\plen(\ell)\ne\plen(z)$ or $\Lout(\ell)\ne\Lout(z)$: if the
preceding run lengths differ, the occurrence adjacent to the shorter run starts
at its left boundary whereas the other is preceded by $a$; if they are equal,
both start at their run boundaries and the outside characters decide left
maximality.

\begin{lemma}\label{lem:locus-representation}
Every valid pair $(X,i,j)$ with $\rho(X)\ge 2$ has a unique representation $(v,\ell,z,K)$ as follows: writing its
longest border as $X=a^KY$ with $Y\ne\eps$, the leaves $\ell$ and $z$, ordered
by $\pos(\ell)<\pos(z)$, represent the two occurrences of $Y$, $v$ is the
(explicit) locus of $Y$ in the sparse suffix tree, the two leaves lie below
distinct children of $v$, and $K=\min(\plen(\ell),\plen(z))$.
\end{lemma}

\begin{proof}
Take the longest border $X$ of the valid pair and let $a^K$ be the \rev{first run
of $X$}.  Write $X=a^KY$.  Then the non-empty suffix
$Y=X[K+1..|X|]$ is uniquely determined and, in both endpoint occurrences of
$X$, the corresponding occurrence of $Y$ starts at a run boundary.  \rev{Since the valid pair satisfies the right-maximality condition, the corresponding occurrences of
$Y$ cannot be followed by the same character.}  Hence their leaves lie below
distinct children of the locus of $Y$; in particular, this locus is explicit.
The two occurrences of $Y$ are preceded by $a$-runs of length at least $K$,
and left maximality implies that one of these runs has length exactly $K$.  If
the two preceding-run lengths are equal, validity requires the outside
characters to be different.  Thus $v$, $\ell$, $z$, and $K$ are uniquely
determined.
\qed
\end{proof}

\subsection{Visibility Reporting}
We use the horizontal visibility relation of Definition~\ref{def:horizontal-visibility}.
For each leaf $y$, let $p_y=(\pos(y),\plen(y))$ be its associated point and
store with $p_y$ a pointer to $y$.  For a node $v$ and a preceding character
$a$, define
\[
P_a(v)=\{p_y:y\in L(v),\ \pchar(y)=a\}.
\]
For two leaves $\ell,z\in L(v)$ with $p_\ell,p_z\in P_a(v)$, let
$K=\min\{\plen(\ell),\plen(z)\}$.  The two occurrences of
$a^K\lab(v)$ are consecutive if and only if no intermediate point in $P_a(v)$
has height at least $K$, or equivalently, if $p_\ell$ and $p_z$ are visible.
Visibility is therefore always tested in the full point set of the 
locus.  Leaves in the same child \rev{cannot satisfy the right-maximality condition}, but they
remain in $P_a(v)$ because they may block visibility between different
children.  Figure~\ref{fig:multi-run-visibility} illustrates this situation.

\begin{figure}[tb]
\centering
\begin{tikzpicture}[
  x=0.62cm,y=0.62cm,
  nodept/.style={circle, draw, fill=white, inner sep=1.6pt},
  leaf/.style={rectangle, draw, rounded corners=1pt, fill=gray!10, inner sep=2pt},
  lightchild/.style={draw=orange!70!black, fill=orange!12, rounded corners=2pt},
  heavychild/.style={draw=blue!70!black, fill=blue!8, rounded corners=2pt},
  point/.style={circle, fill=black, inner sep=1.4pt},
  querypt/.style={circle, fill=orange!80!black, inner sep=1.7pt},
  heavypoint/.style={circle, fill=blue!70!black, inner sep=1.7pt},
  samept/.style={circle, fill=red!70!black, inner sep=1.6pt},
  good/.style={very thick, blue!70!black},
  bad/.style={very thick, red!70!black, dashed}
]

\node[nodept] (v) at (3.25,4.05) {};
\node[above=2pt of v] {$v$};
\node[anchor=east] at (2.30,4.05) {$\lab(v)=Y$};

\node[nodept, lightchild] (u1) at (0.95,2.70) {};
\node[nodept, heavychild] (h)  at (3.25,2.70) {};
\node[nodept, lightchild] (u2) at (5.55,2.70) {};
\draw (v) -- (u1);
\draw[very thick, blue!70!black] (v) -- (h);
\draw (v) -- (u2);
\node[below left=1pt and 1pt of u1] {$u_1$};
\node[below left=1pt and 1pt of h] {$h$};
\node[below right=1pt and 1pt of u2] {$u_2$};
\node[orange!70!black, left=3pt of u1] {query};
\node[blue!70!black, right=3pt of h] {largest};
\node[orange!70!black, right=3pt of u2] {query};

\node[leaf, lightchild] (ell) at (0.20,1.20) {$\ell$};
\node[leaf, lightchild] (y)   at (1.65,1.20) {$y$};
\node[leaf, heavychild] (z)   at (3.25,1.20) {$z$};
\node[leaf, lightchild] (w)   at (5.55,1.20) {$w$};
\draw (u1) -- (ell);
\draw (u1) -- (y);
\draw (h) -- (z);
\draw (u2) -- (w);

\draw[bad, ->]
  (ell.south) .. controls (0.42,0.78) and (1.22,0.82) ..
  (y.south);
\node[red!70!black, fill=white, inner sep=1pt] at (1.48,0.56) {discard};
\draw[good, ->]
  (ell.south) .. controls (0.95,-0.50) and (2.15,-0.78) ..
  (z.south);
\node[blue!70!black, fill=white, inner sep=1pt] at (2.62,-0.26) {report};

\node at (3.20,-1.15) {\textbf{(a) Suffix-tree locus $v$}};

\begin{scope}[shift={(9.7,0)}]
\draw[->] (0,1.20) -- (6.25,1.20) node[right] {$\pos$};
\draw[->] (0,1.00) -- (0,3.40) node[above] {$\plen$};

\node[querypt,label=above left:{$\ell$}] (pell) at (0.80,2.80) {};
\node[samept,label=below right:{$y$}] (py) at (2.15,2.15) {};
\node[point] at (3.05,1.78) {};
\node[heavypoint,label=above right:{$z$}] (pz) at (4.35,2.45) {};
\node[querypt,label=above right:{$w$}] (pw) at (5.35,1.55) {};
\node[point] at (5.95,3.00) {};

\draw[densely dashed, gray] (0.00,2.45) -- (4.35,2.45);
\node[anchor=east] at (-0.12,2.45) {$K$};
\draw[good] (pell) -- (pz);
\draw[bad] (pell) -- (py);

\node at (3.15,-1.15) {\textbf{(b) Full point set $P_a(v)$}};
\end{scope}

\end{tikzpicture}
\caption{Non-unary-border enumeration at a suffix-tree locus $v$. Leaves in non-largest children are queried against the full point set $P_a(v)$. Points in the same child remain as potential blockers, but same-child reports are discarded. A visible leaf $z$ in a different child yields the candidate border $a^K Y$, where $K=\min\{\plen(\ell),\plen(z)\}$.}
\label{fig:multi-run-visibility}
\end{figure}
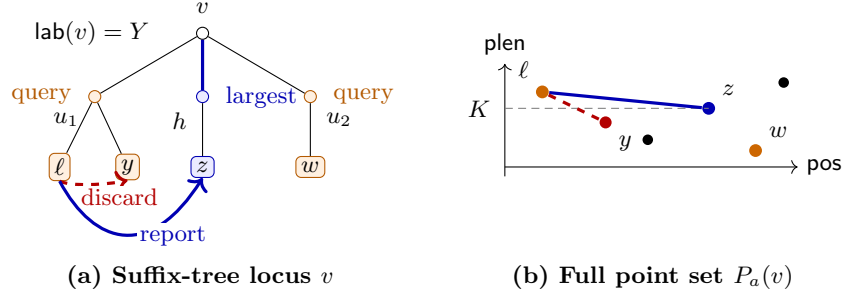

Fixing only $Y$, rather than the whole border $a^KY$, leaves $K$ undetermined,
so one leaf may be visible from several others.  We report all such points
output-sensitively by the following standard three-sided search procedure.

\begin{lemma}[three-sided visibility reporting]\label{lem:visibility-reporting}
Let $P$ be a finite set of points with distinct $x$-coordinates, where each
point $p\in P$ has height $h(p)$.  Suppose that $P$ is stored in a structure
supporting
$\operatorname{Succ}(s,H)=\arg\min_x\{(x,h)\in P:x>s,\ h>H\}$
and the symmetric predecessor query in $O(\log |P|)$ time, with the convention
that either query returns $\mathsf{nil}$ if no such point exists.  Then, for any
query point $q\in P$, all points of $P$ visible from $q$ can be reported in
$O((1+k_q)\log |P|)$ time, where $k_q$ is the output size.
\end{lemma}

\begin{proof}
Consider the points to the right of $q$.  Starting with $H=-\infty$, repeatedly
query $\operatorname{Succ}(x(q),H)$.  If it returns $\mathsf{nil}$, stop.
Otherwise report the returned point $r$; if $h(r)\ge h(q)$, stop, and otherwise
set $H\gets h(r)$ and continue.  Each report raises the threshold to the height
of a visible point, and every skipped point is blocked by a previously reported
point of at least the same height.  Thus there is one query per reported point,
plus one final query.  The left side is symmetric.
\qed
\end{proof}

We implement these queries by McCreight's balanced priority search
tree~\cite{mccreight1985priority}. Fix $N=n+1$ throughout and replace each height $h$ by $N-h$.  \rev{For $H=-\infty$, we omit the height constraint; otherwise $h>H$ becomes}
the three-sided condition $N-h\le N-H-1$, so successor and predecessor reduce to
\textsc{MinXInRectangle} and \textsc{MaxXInRectangle}, respectively, with
$O(\log |P|)$ update and query time and linear space.  Here,
\textsc{MinXInRectangle}$(R)$ (respectively, \textsc{MaxXInRectangle}$(R)$)
returns a point of minimum (respectively, maximum) $x$-coordinate in $P\cap R$,
if one exists.

The point sets are maintained bottom-up and are not rebuilt from scratch.  With
each processed node $x$ we store a dictionary $\mathcal D_x$ indexed by
preceding characters; it is implemented by a balanced binary search
tree~\cite{cormen2022algorithms}.  The entry $\mathcal D_x[a]$ is a priority
search tree storing exactly $P_a(x)$.  For a leaf $y$, $\mathcal D_y$ consists
of $p_y$ in the entry for $\pchar(y)$, unless $\pchar(y)=\bot$.  At an internal
node $v$, choose a child $h$ with maximum leaf-set size, keep $\mathcal D_h$,
and merge every other child dictionary into it to obtain $\mathcal D_v$.

\begin{lemma}\label{lem:outside-visible}
Fix an internal node $v$, one child $u$, and a character
$a$.  Let $q$ be the number of leaves in $L(u)$ with preceding character $a$,
and let $g$ be the number of visible pairs $(\ell,z)$ that pass the
left-maximality test, where $\ell\in L(u)$, $z\in L(v)\setminus L(u)$, and both
leaves have preceding character $a$.  These pairs can be reported in
$O((q+g)\log m)$ time.
\end{lemma}

\begin{proof}
\emph{Correctness.}
Apply Lemma~\ref{lem:visibility-reporting} to $\mathcal D_v[a]$.  For each
queried leaf in $L(u)$, discard reports in the same child and test left
maximality.  Visibility gives consecutiveness; \rev{a report from a different child gives right maximality of the candidate border}; and the final test gives left maximality.

\emph{Complexity.}
Same-child reports form edges of a horizontal visibility graph on the $q$
query leaves and hence contribute $O(q)$ pairs by Lemma~\ref{lem:hvg-linear}.
A left-maximality failure has the same height and outside character as the
queried leaf.  On each side at most one such leaf can be visible, since the
nearest one blocks all farther leaves of the same height.  Hence the total
number of reports is $O(q+g)$, each obtained in $O(\log m)$ time.
\qed
\end{proof}

\subsection{Bottom-Up Enumeration}

We process the sparse suffix tree bottom-up.  At an internal node $v$,
choose a largest child $h$, build $\mathcal D_v$ by keeping $\mathcal D_h$ and
merging all other child dictionaries, and query every leaf in each child
$u\ne h$ against the full set $\mathcal D_v[a]$ for its preceding character
$a$.  Same-child reports are discarded.  A pair between two non-largest
children may be reported from both endpoints, so we keep only the orientation
with the smaller text position; a pair involving $h$ is reported only from its
non-largest endpoint.  After the left-maximality test, the final length test
removes periodic candidates.  Algorithm~\ref{alg:multi-run-locus} gives the
full procedure.

For every suffix-tree node, we store the interval of $\RLESA(T)$ occupied by the leaves in its subtree, and for every leaf we store its index in $\RLESA(T)$. Since the leaves of a subtree are consecutive in suffix-array order, this permits constant-time membership tests in any child subtree and linear-time scanning of all leaves in such a subtree.

For Algorithm~\ref{alg:multi-run-locus}, let \textsc{LeftMax}$(\ell,z)$ abbreviate $\plen(\ell)\ne\plen(z)$ or $\Lout(\ell)\ne\Lout(z)$.

For a valid pair, $X$ is the longest border of its spanned MCS $W$, and hence $\per(W)=|W|-|X|=j-i$. Therefore $\expn(W)<2$ if and only if $|X|<j-i$, which is the final length test in Algorithm~\ref{alg:multi-run-locus}.

\begingroup
\SetAlgoNoLine
\SetAlgoNoEnd
\begin{algorithm}[htb]
\DontPrintSemicolon
\SetKwFunction{ProcessNode}{ProcessNode}
\SetKwProg{Proc}{procedure}{:}{}
\caption{\rev{Non-unary-border enumeration at a node $v$ of the sparse suffix tree.}}\label{alg:multi-run-locus}
\KwIn{An internal node $v$ whose children have already been processed.}
\KwOut{The dictionary $\mathcal D_v$ and the output families generated at $v$.}
\Proc{\ProcessNode{$v$}}{
  let $h$ be a child of $v$ with maximum leaf-set size\;
  form $\mathcal D_v$ by keeping $\mathcal D_h$ and merging every $\mathcal D_u$, $u\ne h$, into it\;
  $Y \gets \lab(v)$\;
  \If{$Y=\eps$}{
    \KwRet $\mathcal D_v$\;
  }
  \ForEach{child $u$ of $v$ with $u \ne h$}{
    \ForEach{leaf $\ell\in L(u)$}{
      $a \gets \pchar(\ell)$\;
      \If{$a\ne\bot$}{
        \ForEach{point $p_z$ visible from $p_\ell$ in $\mathcal D_v[a]$ (with associated leaf $z$)}{
          \If{$z\in L(u)$}{
            \KwContinue\;
          }
          \If{$z\notin L(h)$ and $\pos(\ell)>\pos(z)$}{
            \KwContinue\;
          }
          \If{\textnormal{\textsc{LeftMax}}$(\ell,z)$}{
            $K \gets \min(\plen(\ell),\plen(z))$\;
            $i \gets \min(\pos(\ell),\pos(z))-K$\;
            $j \gets \max(\pos(\ell),\pos(z))-K$\;
            \If{$|a^KY|<j-i$}{
              \textbf{output} $\mathsf{S}(i,j+|a^KY|-1)$\;
            }
          }
        }
      }
    }
  }
  \KwRet $\mathcal D_v$\;
}
\end{algorithm}
\endgroup

Each visibility loop uses Lemma~\ref{lem:visibility-reporting}; the output
$\mathsf S(i,j+|a^KY|-1)$ is the singleton family containing the MCS occurrence
spanned by the two border occurrences.

\subsection{Correctness and Complexity}

\begin{lemma}\label{lem:multi-run-correct}
The above enumeration outputs exactly the singleton families corresponding
to valid pairs $(X,i,j)$ with $\rho(X)\ge 2$ and
$\expn(\spanstr(X,i,j))<2$, without duplication.
\end{lemma}

\begin{proof}
\emph{Soundness.}
Every output pair consists of leaves below distinct children of the 
locus $v$ of $Y$, so \rev{the candidate border is right-maximal}.  Visibility
is tested in the full point set $P_a(v)$ and therefore excludes every
intermediate occurrence of $a^KY$.  The explicit test on $\plen$ and $\Lout$
gives left maximality.  Hence the pair is valid.  Since $Y\ne\eps$, the border
has at least \rev{two runs}, and the final length test \rev{removes periodic spans}.

\emph{Completeness and uniqueness.}
\rev{Take a valid pair of this type with a non-periodic span.}  By
Lemma~\ref{lem:locus-representation}, its longest border is uniquely
$X=a^KY$, and its two leaves lie below distinct children of the locus $v$ of
$Y$.  At least one of these children is not the chosen largest child $h$.
Therefore the pair is reported from a queried endpoint.  If exactly one
endpoint lies in $h$, it is kept from the other endpoint.  If both endpoints lie
in non-largest children, it is reported twice and exactly the orientation with
the smaller text position is kept.  Thus the pair is output exactly once after
passing the left-maximality and non-periodicity tests.
\qed
\end{proof}

\begin{lemma}\label{lem:multi-run-time}
The non-unary Type~4 families can be enumerated in $O(m\log^2 m+f\log m)$
time, where $f$ is their number, using $O(m)$ working space.
\end{lemma}

\begin{proof}
\emph{Queries.}
A leaf is queried at $v$ only when it belongs to a non-largest child $u$.
Since $u$ is not larger than the chosen largest child,
$|L(u)|\le |L(v)|/2$.  Hence, whenever a fixed leaf is queried, the size of
the subtree containing it at least doubles when moving from $u$ to $v$.
Consequently each leaf is queried at $O(\log m)$ nodes, for a total of
$O(m\log m)$ leaf queries.

\emph{Structures.}
At each internal node $v$, we keep the dictionary of a largest child $h$ and
merge the dictionaries of all other children into it.  For an entry
$\mathcal D_u[a]$ of a non-largest child $u$, if the current dictionary has no
entry for $a$, its priority search tree is reused directly as the new entry;
otherwise, its points are inserted individually into the existing priority
search tree for $a$.  Every such individual insertion comes from a
non-largest child, so the size of the containing subtree at least doubles
before the same point can be inserted again, and hence each point is inserted
at most $O(\log m)$ times. The same doubling argument applies when an entire priority search tree is reused as a new dictionary entry, charging this operation to an arbitrary point in that entry.  Thus there are $O(m\log m)$ dictionary
operations and point insertions in total, each costing $O(\log m)$ time in the
worst case, for a total of $O(m\log^2 m)$ time.  Since the dictionary of $h$
is reused as $\mathcal D_v$ and the other child dictionaries are discarded
after their entries have been merged, each point is stored in exactly one
current dictionary, and the working space is $O(m)$.

\emph{Reports.}
We apply Lemma~\ref{lem:outside-visible} separately to every non-largest child and preceding character. Over all these applications, the $q$ values sum to at most the total number of leaf queries, and hence to $O(m\log m)$. A pair between a largest and a non-largest child is reported from one queried endpoint, while a pair between two non-largest children is reported from at most two endpoints. \rev{Thus non-periodic spans contribute $O(f)$ to the sum of the $g$ values, while periodic spans correspond to the $O(m)$ Type~5 MCSs and contribute $O(m)$.} Hence the $g$ values sum to $O(f+m)$. Therefore the total reporting time is $O((m\log m+f+m)\log m)=O(m\log^2 m+f\log m)$, and the claimed bound follows.
\qed
\end{proof}

\section{Putting the Bounds and Enumeration Together}\label{sec:main-enumeration}

We are ready to prove Theorem~\ref{thm:main}.
\begin{proof}
The classification of Definition~\ref{def:types} covers every MCS occurrence
of length greater than one.  Lemma~\ref{lem:type5} handles Types~1 and~5,
Lemma~\ref{lem:unary-enum} handles Types~2 and~3, and
Lemma~\ref{lem:multi-run-time} handles Type~4; length-one MCS occurrences are
reported trivially.  Together with the preprocessing of
Section~\ref{sec:sparse-st}, these algorithms enumerate a representation
$\mathcal F$ in $O(m\log^2m+|\mathcal F|\log m)$ time and $O(m)$ working
space, since $\sigma\le m$ and $\alpha(m)=O(\log m)$.
Theorem~\ref{thm:family-size} gives $|\mathcal F|=O(m^2)$.
\qed
\end{proof}

\section{Conclusions and Open Problems}\label{sec:conclusions}

We gave a compact representation of all MCS occurrences in a run-length
encoded string.  Its worst-case size is $O(m^2)$, and this bound is tight for
the family representation introduced in this paper.  We also gave an
output-sensitive algorithm that computes the representation in
$O(m\log^2 m+|\mathcal F|\log m)$ time using $O(m)$ working space.  The main
structural ingredient is to represent an MCS occurrence by two consecutive
occurrences of its longest border and to exploit the RLE structure of that
border.

Several questions remain open.  First, our quadratic lower-bound construction
uses an alphabet whose size grows with $m$.  Is the $O(m^2)$ bound still tight
over a constant-size, or even binary, alphabet?  Second, can one remove one or
both logarithmic factors from the enumeration time, for example obtaining
output-linear time after near-linear preprocessing?  Finally, our lower bound
applies to the family representation defined in Section~\ref{sec:classification}.
It remains open whether a fundamentally different compact representation can
encode all MCS occurrences in $o(m^2)$ space.

\clearpage

\bibliographystyle{splncs04}
\bibliography{ref}

\end{document}